\documentclass[conference]{IEEEtran}
\usepackage[pdftex]{graphicx}
\ifCLASSINFOpdf
\DeclareGraphicsExtensions{.pdf,.jpeg,.png}
\usepackage{amsmath}
\usepackage{amsthm}
\usepackage{algorithm}
\usepackage[noend]{algpseudocode}
\usepackage{array}
\usepackage{amsfonts}
\usepackage{verbatim}
\usepackage[utf8]{inputenc}
\usepackage[english]{babel}
\usepackage{multirow}
\newtheorem{theorem}{Theorem}

\newtheorem{lemma}{Lemma}
\newtheorem{defn}{Definition}
\usepackage[table]{xcolor}

\begin{document}
\title{Assigning Optimal Integer Harmonic Periods to Hard Real-Time Tasks}
\author{\IEEEauthorblockN{Anand Bhat and Ragunathan (Raj) Rajkumar}
	\IEEEauthorblockA{Electrical and Computer Engineering,\\
	Carnegie Mellon University,\\
	Pittsburgh, PA 15213.\\
	Email: anandbha, rajkumar@andrew.cmu.edu}
}
\maketitle
\begin{abstract}
Selecting period values for tasks is a very important step in the design process of a real-time system, especially due to the significance of its impact on system schedulability. It is well known that, under RMS, the utilization bound for a \emph{harmonic} task set is 100\%. Also, polynomial-time algorithms have been developed for response-time analysis of harmonic task sets. In practice, the largest acceptable value for the period of a task is determined by the performance and safety requirements of the application. In this paper, we address the problem of assigning harmonic periods to a task set such that every task gets assigned an integer period less than or equal to its application specified upper bound and the task utilization of every task is less than 1. We focus on integer solutions given the discrete nature of time in real-time computer systems.

We first express this problem of assigning harmonic periods to a task set as a \emph{discrete piecewise optimization} problem. We then present the \emph{``Discrete Piecewise Harmonic Search''} (DPHS) algorithm that outputs an optimal harmonic task assignment. We then define conditions for a metric to be \emph{rational} for harmonization. We show that commonly used metrics like, the total percentage error (TPE), total system utilization (TSU), first order error (FOE), and maximum percentage error (MPE), are rational. We next prove that the DPHS algorithm finds the optimal feasible assignment, if one exists, for these rational metrics. We apply the DPHS algorithm to harmonize task sets used in real-world applications to highlight its benefits. We compare the performance of the DPHS algorithm against a brute-force search and find that the DPHS searches up to 94\% fewer task sets than the brute-force search that obtains the optimal solution.
\end{abstract}

\section{Introduction}\label{Intro}

Real-time systems closely interact with the environments in which they are deployed \cite{CPS}. Owing to the recurring nature of events in such environments, a periodic task model is widely used in these systems. For example, in the autonomous driving context, a car needs to sense the environment, perform calculations and control actuators on a periodic basis. This is accomplished by using tasks that run periodically, using a preemptive real-time scheduling policy like \emph{Rate-Monotonic Scheduling} (RMS)  \cite{Liu}. The selection of these periods plays a very important role in the design, analysis, and schedulability of a real-time system.

The selection of task periods is driven by the safety and performance specifications of a real-time application \cite{Ryu}, \cite{Tools}. This choice naturally has a direct impact on system schedulability. For example, a feedback control application can perhaps produce very accurate control if it runs at a very high frequency, i.e., if the period assigned to the task running the feedback control application is small. However, since smaller periods mean higher CPU utilization, system schedulability is reduced \cite{Seto}. It has been shown that the exact schedulability analysis for the RMS policy is an NP-Complete problem \cite{Eisenbrand}, unless the periods are harmonic \cite{Han}, i.e., every period in the task set is an integer multiple of its shorter periods. Also, polynomial-time solutions exist for the response-time analysis of systems with harmonic task sets \cite{Bonifaci}. 

There are also several considerations related to the periodic nature of tasks. The nature of RM schedulability \cite{Liu} means that a CPU cannot be completely utilized, i.e., reach 100\% utilization, unless the task periods are harmonic. However, in Section \ref{SchedImpact}, we show that, subject to the constraint of having periods less than or equal to the original periods after transformation, a non-schedulable task set cannot be made schedulable by making the periods harmonic. However, having harmonic periods allows for phase optimizations that reduce communication latencies \cite{Gerber} and also enables energy saving optimizations \cite{Sandeep}. Finally, harmonic task sets play an important role in reducing the complexity in the design of distributed time-triggered embedded systems \cite{Kopetz}. In practice, harmonic task sets are widely used in real-time safety-critical systems like automobiles and avionics \cite{Dodd, Busquests, Avionics, Avionics2}.

In this paper, we consider a set of independent periodic tasks with application-specified period values. Considering the benefits of harmonic task sets detailed above, the goal of this paper is to assign a set of integer harmonic periods to the input task set while optimizing a rational metric as defined in Section \ref{SecRationalMetrics}. Some examples of rational metrics are total percentage error (TPE), total system utilization (TSU), first order error (FOE), and maximum percentage error (MPE). We ensure that every period assigned is less than or equal to the original application-specified period and that the task utilization of every task is less than 1 (where applicable). This constraint is intended to satisfy the application safety and performance specification. We consider only integer periods owing to the discrete nature of time in real-time computer systems.

In this paper, we will then present a framework to represent integer harmonic task sets. We show that every integer period from the task set can be described in terms of a multiplier, a base and an exponent. We derive the lower and upper bounds for these parameters for a given task set. The goal is then to find optimal values of the multiplier and the base such that every period assigned is less than or equal to the original application-specified upper bound. Where applicable, we strive to ensure that the utilization of every task in the system is less than 1, i.e. every task is assigned a period greater than its worst-case execution time, ensuring that it can complete execution. We also describe a brute-force approach, which searches through all possible values for the multiplier and the base to find the globally optimal solution. We then represent the harmonization problem as a discrete piecewise optimization problem. For a given multiplier, we identify bases that result in local optima for a given metric. We describe our ``Discrete Piecewise Harmonic Search'' (DPHS) algorithm that identifies these local optima and searches through them to obtain the global optimal. We compare the performance of the DPHS algorithm against the brute-force approach in terms of the number of task sets checked for optimality and running time.

The rest of this paper is organized as follows. We present related work in Section \ref{RelatedWork}. We describe the schedulability impact of harmonization in Section \ref{SchedImpact}. We present the problem statement and describe our system model in Section \ref{systeMBOdel}. We present a mathematical representation of harmonic integer sets in terms of a multiplier, a base, and an exponent in Section \ref{MathRep}. We derive lower and upper bounds on these parameters in Section \ref{Bounds}. We present the brute-force approach to the problem in Section \ref{BruteForce}, followed by the DPHS algorithm in section \ref{Pruned}. In Section \ref{Application}, we discuss the usage scenarios for our approach by applying the DPHS algorithm to real-world task sets. In Section \ref{Evaluation}, we evaluate the performance of the DPHS algorithm against the brute-force approach. We also discuss special properties of the algorithm. In Section \ref{Conclusion}, we summarize our findings.     

\section{Related Work}\label{RelatedWork}
There has been significant work in the field of task period selection in hard real-time systems for optimizing various use cases.  For example, in \cite{Control1}, Henriksson et al. attempt to find an optimal period assignment to distribute computing resources between tasks. In \cite{Control2}, Cervin et al. try to optimize the system performance of a control system. Neither, however, focuses on creating harmonic task sets.

In \cite{Han}, Han et al. have presented the \emph{Sr} and the \emph{DCT} algorithms which attempt to create a harmonic task set to verify the schedulability of a given input non-harmonic task set. Since there exist linear-time solutions to check the schedulability of systems with harmonic task sets, both algorithms attempt to find a harmonic task set such that, every period assigned is less than or equal to the original period and the utilization of every task is less than 1. Schedulability of the original task set is conservatively inferred from the schedulability of the harmonic task set. However, the Sr algorithm assigns artificial periods which are always multiples of 2, whereas the DCT algorithm creates a harmonic series using each term in the original task set. Unlike our work presented in this paper, neither algorithm attempts to find the optimal assignment. Our solution in this paper can be applied readily to solve the above problem. Similarly, in \cite{Min}, harmonic deadline assignments were created using least common multipliers (LCM) of the original task periods, to derive a utilization-based feasibility test for systems with composite deadlines, but an optimal assignment is not attempted to be created.

In \cite{Nasri} and \cite{Fohler}, Nasri et al. address a problem similar to the one addressed in this paper, but accept a range of period values for each task as input and target a feasible real number solution, whereas this paper looks to find an optimal integer solution, where chosen values are below specified thresholds.

\section{Schedulability Impact of Harmonization}\label{SchedImpact}

As highlighted in Section \ref{Intro}, it is tempting to think that harmonization can improve schedulability. However, in this section, we show that a non-schedulable task set cannot be made schedulable by harmonization, constrained to having the periods less than or equal to the original periods for rate-monotonic scheduling policy with implicit deadlines.
\begin{theorem}\label{TheononSched}
	A non-schedulable task set cannot be made schedulable by transforming the periods to values less than or equal to the original periods for rate monotonic scheduling policy with implicit deadlines.
\end{theorem}
\begin{proof}
	For a task set $\tau = \{\tau_{1}, \tau_{2}, ... \tau_{n}\}$ with tasks arranged with increasing periods i.e $T_{i} \le T_{i+1}$ $\forall$ $i$ from 1 to $n$, the schedulability under rate-monotonic scheduling can be determined using the response time test, which is as follows,
	
	\begin{equation}
	\begin{aligned}
	& a^{i}_{k+1} = C_{i} +  \sum_{j=1}^{i-1} \lceil a^{i}_{k}/T_{j} \rceil C_{j} \\
	& a^{i}_{0} =  \sum_{j=1}^{i} C_{j}
	\end{aligned} 
	\end{equation}
	
	where,  $a^{i}_{k}$ represents an estimate of the response time for task $\tau_{i}$.
	
	This implies that, for a non-schedulable task set at least one task remains unschedulable, i.e. its response time is greater than its deadline. Hence, if $\tau_{m}$ ($m = 1$ to $n$) is non-schedulable we have,
	
	\begin{equation}
	\label{ResponseTimeTest}
	\begin{aligned}
	& a^{m}_{k+1} = C_{m} +  \sum_{j=1}^{m-1} \lceil a^{i}_{k}/T_{j} \rceil C_{j} > T_{m}\\
	\end{aligned} 
	\end{equation}
	
	From the above equation, we see that the response time estimate depends only on the periods of tasks with $i < m$, i.e., the tasks with smaller periods and hence higher priority under RMS. Since $T_{j}$ is part of the denominator transforming periods values from $T_{i}$ to $T'_{i}$ such that $T'_{i} < T_{i}$ can only result in larger response times. 
\end{proof}

Hence, from the above theorem, even by harmonization, a non-schedulable task set cannot be made schedulable by transforming the periods to values less than or equal to the original periods for rate monotonic scheduling policy with implicit deadlines. Conversely, using identical arguments from Theorem \ref{TheononSched}, we can show that the schedulability of a harmonic task set can only improve, if on transformation the period values are allowed to be greater than the original periods.  
\section{System Model and Problem Definition}
\label{systeMBOdel}
In this section, we describe our system model and present our problem statement.

\subsection{System Model}
We consider a hard real-time task set $\tau = \{\tau_{1}, \tau_{2}, \tau_{3},\ldots, \tau_{n}\}$, which consists of $n$ independent periodic tasks. Each task $\tau_{i}$ is described as $(C_{i}, T_{i}, D_{i})$, where $C_{i}$ ($ C_{i} \in \mathbb{R}_{>0}$) represents the worst-case execution time of the task, $T_{i}$ ($ T_{i} \in \mathbb{N}_{>0}$) is the period of the task and $D_{i}$ ($ D_{i} \in \mathbb{N}_{>0}$) is the deadline by which the task is expected to complete execution. We assume implicit deadlines, i.e., $D_{i} = T_{i}$. We also assume that the tasks in the task set are ordered by non-decreasing periods, i.e., $T_{1} \leq T_{2} \leq \ldots \leq T_{n}$.\\

We define a \emph{harmonic set} as a set of numbers where every number is an integer multiple of every smaller number from the set. A task set in which the set of all task periods form a harmonic set is referred to as a \emph{harmonic task set} and the set of task periods is referred to as a \emph{harmonic period set}.

\subsection{Problem Statement}
Given a task set $\tau$, generate an \emph{optimal} harmonic task set $\tau'$ with integer periods (i.e. $T'_{i} \in \mathbb{N}_{>0}$), if one exists, such that,
\begin{enumerate}
	\item The period of every task in $\tau'$ is an integer less than or equal to its corresponding period in $\tau$, i.e., $T'_{i} \leq T_{i}$.
	\item Every task continues to remain schedulable, i.e., $C_{i} \leq T'_{i}$.
\end{enumerate}

A resulting task set is said to be \emph{optimal} if it optimizes the selected rational metric as defined in Section \ref{MathRep}.\\

\subsection{Mathematical Representation of Harmonic Sets}\label{MathRep}
We now derive a mathematical representation for our problem definition.

First, consider an integer harmonic set $\zeta = \{\zeta_{1}, \zeta_{1}, ..., \zeta_{n}\}$, with $n$ numbers.  Since every $\zeta_{i}$ is an integer multiple of a smaller number in the set, every $\zeta_{i}$ can be represented as
\begin{equation}
\begin{aligned}
& \zeta_{i} = m * b^{x_{i}}\\ 
\end{aligned} 
\end{equation}
where, we refer to $m \in \mathbb{I}$ as the \emph{multiplier}, $b \in \mathbb{N}_{>0}$ as the \emph{base} and $x_{i} \in \mathbb{N}_{\ge 0}$ as the exponent. 
That is, every period $T'_{i}$ in the harmonic period set can be represented as
\begin{equation}
T'_{i} = m * b^{x_{i}}
\end{equation}

Given that every $T'_{i}$ represents a task period, the multiplier is a positive integer, i.e.,  $m \in \mathbb{N}_{>0}$.

Example: Consider the integer harmonic set {10, 20, 40}. These numbers correspond to $m = 10, b=2, x = \{0,1,2\}$.
Since every number in the harmonic period set is less than or equal to the corresponding number in the original task set, we have,
\begin{equation}
\begin{aligned}
& T_{i} = T'_{i} + K_{i}
& \text{where } K_{i} \in \mathbb{R}_{=>0}\\
& T_{i} = m * b^{x_{i}} + K_{i}
\end{aligned} 
\end{equation}

A real-time systems engineer may be interested in reduced total system utilization, or having individual periods not to be too small compared to the original periods depending on the application. We use a general cost function to capture a variety of such metrics.
\begin{figure}[t] 
	\centering
	\includegraphics[page=6, width=8.5cm, height=4cm]{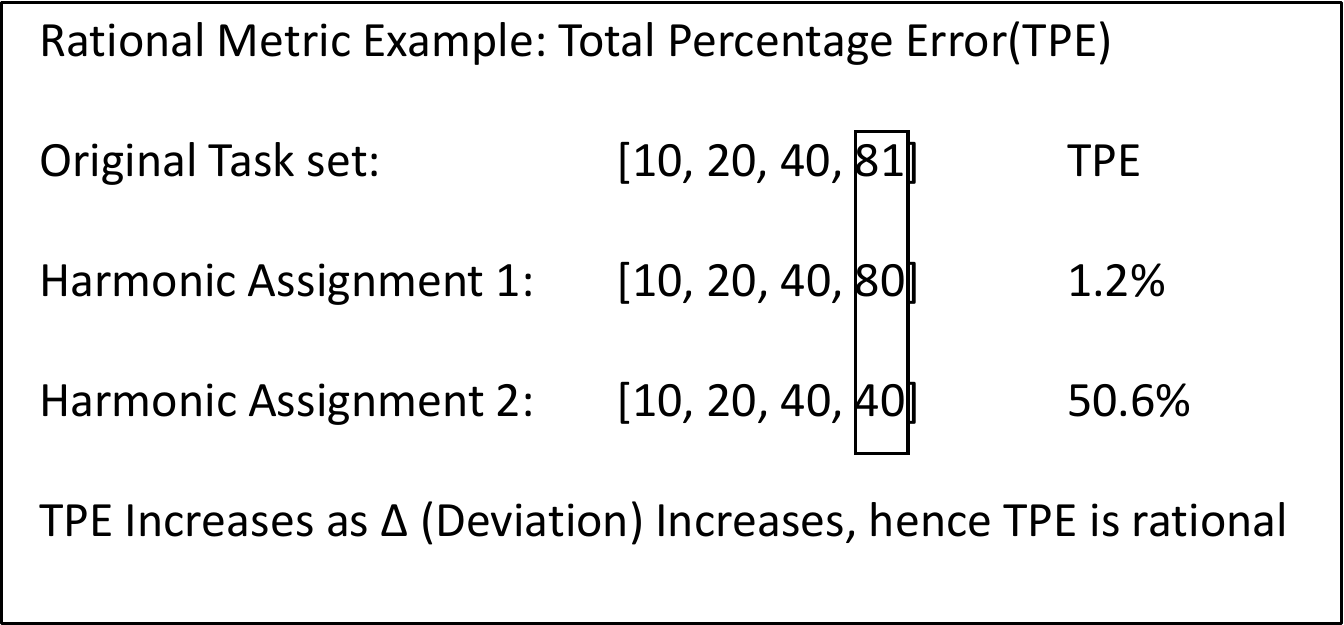}
	\caption{Rational Metric Example}
	\label{RationalExample}
\end{figure}
\subsection{Rational Metrics}\label{SecRationalMetrics}
\begin{defn}
A metric for harmonization is said to be rational if the value of its cost function increases as the deviation of any one period in the resultant harmonic task set from the original task set increases.
\end{defn}

Our  goal is to find $m^{*}$ and $b^{*}$, which represent the values of the multiplier and the base optimizing a given rational metric, as defined above. Most of the commonly used error metrics are rational. Below are a few examples. 

\begin{enumerate}
	
	\item Total percentage error (TPE), i.e., 
	\begin{equation*}
	\begin{aligned}
	& hcf(m,b,x) = \underset{T'}{\text{Minimize}}
	& &  \sum_{i=1}^{n} (T_{i} - T'_{i})/T_{i} \\
	& \Rightarrow hcf(m,b,x)= \underset{m, b}{\text{Minimize}}
	& &  \sum_{i=1}^{n} (T_{i} - (m * b^{x_{i}}))/T_{i} \\
	\end{aligned}
	\end{equation*}
	
	\item Total system utilization (TSU), i.e., 
	\begin{equation*}
	\begin{aligned}
	& hcf(m,b,x) = \underset{T'}{\text{Minimize}}
	& & \sum_{i=1}^{n} (C_{i}/T'_{i}) \\
	& \Rightarrow hcf(m,b,x)= \underset{m, b}{\text{Minimize}}
	& & \sum_{i=1}^{n} (C_{i}/(m * b^{x_{i}})) \\
	\end{aligned} 
	\end{equation*}
	
	\item First order error (FOE) i.e., 
	\begin{equation*}
	\begin{aligned}
	& hcf(m,b,x) = \underset{T'}{\text{Minimize}}
	& & \sum_{i=1}^{n} (T_{i} - T'_{i}) \\
	& \Rightarrow hcf(m,b,x)= \underset{m, b}{\text{Minimize}}
	& & \sum_{i=1}^{n} (T_{i} - (m * b^{x_{i}})) \\
	\end{aligned}
	\end{equation*}
	
	\item Maximum percentage error (MPE),
	\begin{equation*}
	\begin{aligned}
	& hcf(m,b,x) = \underset{T'}{\text{Minimize}}
	& & \underset{i}\max [(T_{i} - T'_{i})/T_{i}]. \\
	& \Rightarrow hcf(m,b,x) = \underset{m, b}{\text{Minimize}}
	& & \underset{i}\max [(T_{i} - (m * b^{x_{i}}))/T_{i}]. \\	\end{aligned} 
	\end{equation*}
	
\end{enumerate}

Figure \ref{RationalExample} illustrates rationality of the TPE metric. The rationality of the other metrics shown below can be inferred in the same way. 

The choice of the cost function is application-dependent. For example, a resource-constrained system would optimize for reduced system utilization whereas, an application with relatively stricter period requirements would prefer to minimize the maximum percentage error from the original task set. This list of cost functions above is not exhaustive, additional details will be provided in Section \ref{Pruned}.

\section{Harmonic Brute-Force Search Algorithm}\label{BruteForce}
\subsection{Bounds on the multiplier, the base, and the exponent}\label{Bounds}
In this section, we present some properties of the period values we are working with in terms of the mathematical representation from section \ref{MathRep}.
\begin{lemma}
	\label{lemma1}
	Given a set of periods $T$, the range of legal or valid values that the multiplier $m$ of the output harmonic period set ranges from $1$ to $T_{1}$.
\end{lemma}
\begin{proof}
	Every element of the output harmonic period set is constrained to be less than or equal to the corresponding element in the input period set. Consider the first element $T'_{1}$, which can be represented as follows,
	\begin{equation*}
	\begin{aligned}
	& T'_{1} = m * b^{x_1} \leq T_{1}\\
	& m_{max} \leq T_{1} \text{ when } b = 1 \text{ or }  x_1 = 0
	\end{aligned} 
	\end{equation*}
	The value of the multiplier is maximized when the second term, $b^{x_1}$,  is $1$, i.e., either $b = 1$ or $x_1 = 0$. Hence, the maximum value of multiplier is $T_{1}$ and, since the multiplier is a positive integer, the minimum value for the multiplier is 1. 
\end{proof}
\begin{lemma}
	\label{lemma2}
	Given a set of periods $T$, the base $b$ of the output harmonic period set, for a given multiplier $m$, ranges from $1$ to $\lfloor T_{n}/m \rfloor$.
\end{lemma}
\begin{proof}
	We have, $\forall$ $i$, $T'_{i} \le T_{i}$. Consider the last element $T'_{n}$, which can be represented as follows,
	\begin{equation*}
	\begin{aligned}
	& T'_{n} = m * b^{x_n} \leq T_{n}\\
	& b^{x_n} \leq T_{n}/m \text{ since } m \in \mathbb{N}_{>0} \\
	& b_{max} \leq T_{n}/m \text{ when } x_{n} = 1 \\
	& \Rightarrow b_{max} = \lfloor T_{n}/m \rfloor \text{ since } b \in \mathbb{N}_{>0} 
	\end{aligned} 
	\end{equation*}
	For a given multiplier $m$, the value of the base is maximized when $x_n = 1$. Hence, the maximum value of multiplier is $\lfloor T_{n}/m \rfloor$ and since the base is a positive integer, the minimum value for the base is 1.   
\end{proof}

\begin{lemma}
	\label{lemma3}
	Given an input set of periods $T$, the exponent $x_{i}$ for each term of the output harmonic period set ranges from 0 to $\lfloor log_{2}(T_{i}/m) \rfloor$. For $b = 1$ the exponent is irrelevant (i.e., any value of $x_i \geq 0$ has the same effect).
\end{lemma}
\begin{proof}
	We have, $\forall$ $i$, $T'_{i} \le T_{i}$. That is,\
	\begin{equation*}
	\begin{aligned}
	& T'_{i} = m * b^{x_i} <= T_{i}\\
	& {x_i} \leq log_{b}(T_{n}/m)\\ 
	& \text{ since } b \in \mathbb{N}_{=>1} \text{ and } m \in \mathbb{N}_{>0}, \text {we have} \\
	& (x_i)_{max} \leq log_{2}(T_{n}/m) \text{ when } b = 2 \\
	& (x_i)_{max} = \lfloor log_{2}(T_{n}/m) \rfloor \text{ since } x_i \in \mathbb{N}_{=>0} 
	\end{aligned} 
	\end{equation*}
	As the base decreases, the value of the exponent increases. For a given multiplier $m$, the value of the exponent is maximized when the value of the base is minimum i.e., in this case, $b=2$ since $b \in \mathbb{N}_{>1}$. Hence, the maximum value of the exponent is $\lfloor log_{2}(T_{n}/m) \rfloor$ for $b \in \mathbb{N}_{>1}$ and, since the exponent is a non-negative integer, the minimum value for the exponent is 0. When $b=1$, the second term $b^{x_i}$ will always result in a value of 1, i.e. every $T'_{i} = m$, so any value of $x_i >=0$ has the same effect.
\end{proof}

\begin{figure}[t] 
	\centering
	\includegraphics[trim =30mm 0mm 0mm 0mm, clip, width=10.5cm, height=4.5cm]{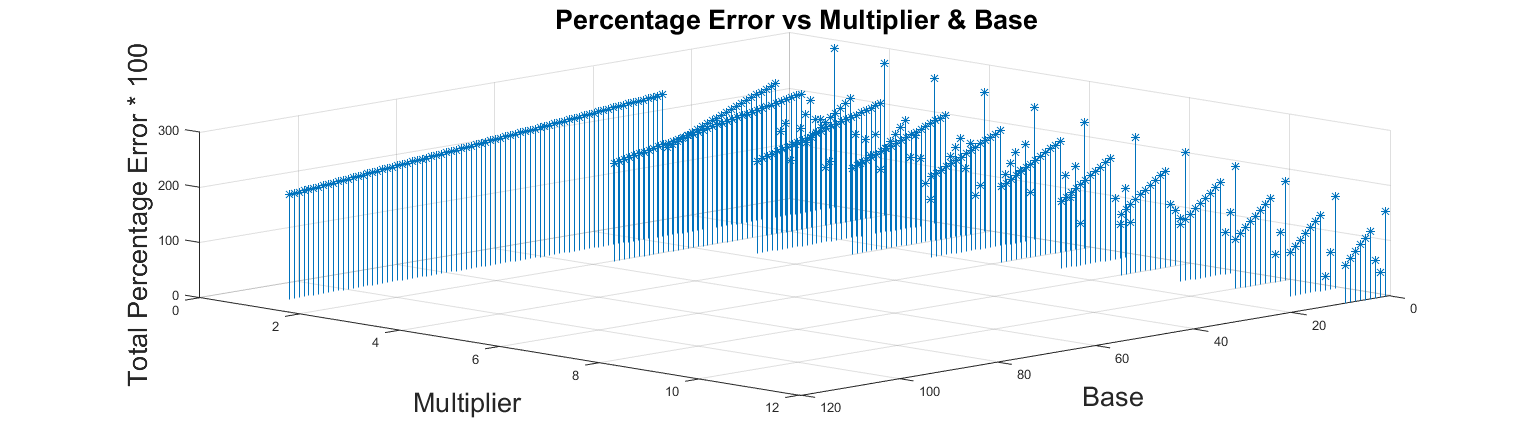}
	\caption{Brute-Force Search Plot}
	\label{BruteForceSearchPlot}
\end{figure}

\begin{figure}[t] 
	\centering
	\includegraphics[width=7cm, height=6.5cm]{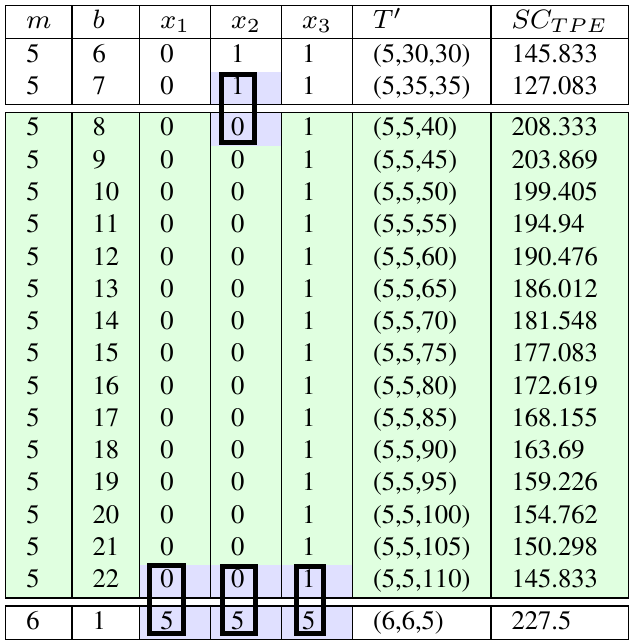}
	\caption{Brute-Force Search tables (Best Viewed in Color)}
	\label{BruteForceSearchTables}
\end{figure}

\begin{algorithm} [t]
	
	\caption{$FindClosestHarmonicSeries$}
	\begin{algorithmic}[1]
		\Procedure{findClosestHarmonicSeries}{}
		\State $m \gets \text{multiplier}$
		\State $b \gets \text{base}$
		\State $\tau \gets \text{input task set}$
		\State $T \gets \text{period set from } \tau$
		\For{each $\tau_{i} \text{ in } \tau }$
		\State $x_i = \lfloor log_{b}(T_{i}/m) \rfloor$ 
		\State $T'_{i} = m * b^{x_{i}}$
		\EndFor
		\State \textbf{return} $T'$
		\EndProcedure
	\end{algorithmic}
\end{algorithm}

\begin{algorithm} [t]
	
	\caption{Harmonic Brute-Force Search Algorithm}
	\begin{algorithmic}[1]
		\Procedure{bruteForceSearch}{}
		\State $\tau \gets \text{input task set}$
		\State $T \gets \text{period set from } \tau$
		\State $hcf \gets \text{cost function to minimize}$
		\State $m_{max} = T_{1}$
		\State $error = numMax$
		\For{each multiplier $m \text{ from  } 1 \text{ to } m_{max}$ }
		\State $b_{max} = \lfloor T_{n}/m \rfloor$
			\For{each base $b \text{ from  } 1 \text{ to } b_{max}$ }
				\State $T'_{temp} = findClosestHarmonicSeries(\tau, m , b) $
				\State $errorTemp = findError(hcf, \tau, T'_{temp})$
				\If {$errorTemp < error$  \& $C_{i} <= T'_{temp}$}
					\State $T' = T'_{temp}$
				\EndIf
			\EndFor
		\EndFor
		\State \textbf{return} $T'$ \Comment{Return the closest feasible harmonic period set or null to indicate infeasible result}
		\EndProcedure
	\end{algorithmic}
\label{Harmonic Brute-Force Search Algorithm}
\end{algorithm}

From Lemma \ref{lemma1} and Lemma \ref{lemma2}, we have the bounds on the multiplier and the base of the output harmonic period set $T'$ for an input period set $T$. We need to iterate over all combinations of the multiplier and base, and find the optimal values $m^{*}$ and $b^{*}$ for chosen cost function from Section \ref{MathRep}.\\ 

In each iteration, the values of $m$ and $b$ are known, the only unknown is $x$ which is efficiently calculated as the highest exponent to which the base $b$ can be raised before $T'_{i}$ becomes greater than $T_{i}$. This is detailed in the \emph{findClosestHarmonicSeries} function in Algorithm 1.\\ 

As an example, consider a simple task set $\tau = \{(2,12), (3,35), (2,112)\}$ as an input to the Harmonic Brute-Force Search Algorithm. Suppose the cost function of interest is minimizing overall percentage error, i.e., the first cost function detailed in Section \ref{MathRep}. Figure \ref{BruteForceSearchPlot} plots the scaled total percentage error for every combination of the multiplier and base as calculated by the Harmonic Brute-Force Search Algorithm. From Lemma \ref{lemma1}, the multiplier values range from 1 to 12 since $T_{1} = 12$. For each multiplier, the bases have a different range as given by Lemma \ref{lemma2} and represented in the plot. For a given multiplier and a range of bases, the plot shows a constant decreasing trend in the total percentage error. Also, the trend repeats for different multipliers and base ranges. Figure \ref{BruteForceSearchTables} highlights the values corresponding to one such trend, i.e., for $m = 5$ and $b$ from 8 to 22. The figure shows the output harmonic period set and the exponent values corresponding to each term in this set, for a given multiplier and base combination. It also shows the corresponding scaled cumulative percent error  ($SC_{TPE}$) of the output harmonic period set with respect to the original task set $T$. Notice that, for this range the values of $x_{i}$ remain constant across all entries, i.e $x = \{0,0,1\}$. As can be seen only the last entry in this range is of interest in the search for an optimal assignment. This is the key insight on which the Discrete Piecewise Search Algorithm described in the next section is based on. Algorithm 2  details the brute-force search approach. Notice that the brute-force algorithm also checks the feasibility of the solution, i.e., $C_{i} \leq T'_{i}$, and it returns a solution if one exists, else returns null.

\begin{figure}[t] 
	\centering
	\includegraphics[page=5, width=8.5cm, height=4.5cm]{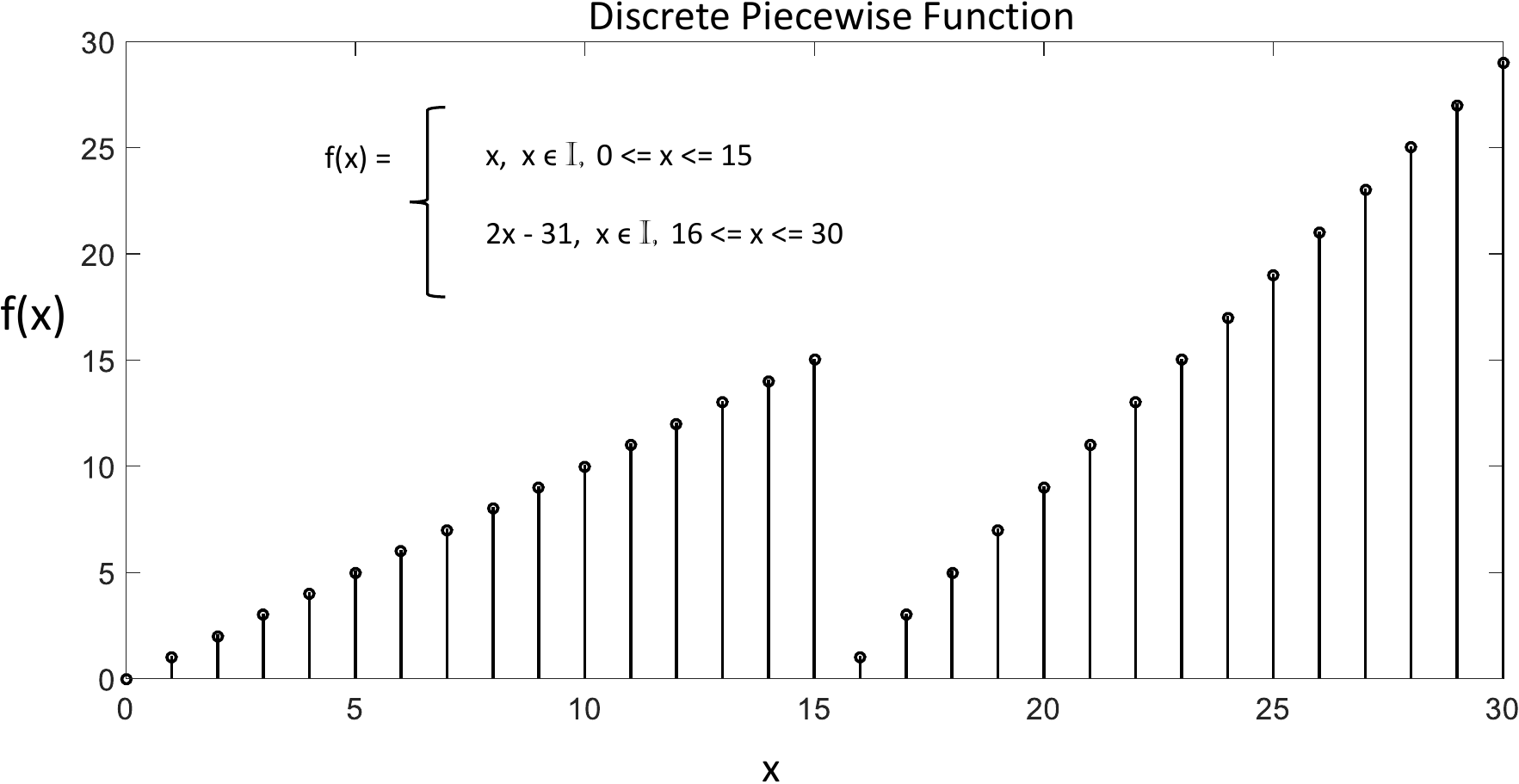}
	\caption{Discrete Piecewise Function Example}
	\label{DPFExample}
\end{figure}

\section{Discrete Piecewise Search Algorithm}\label{Pruned}
\subsection{Discrete Piecewise Optimization}\label{SecDPO}
In this section, we present the discrete piecewise optimization approach.
\begin{defn}
	A function is said to be a discrete piecewise function if it can be represented as follows,
	
	\[ 
	f(x) = 
	\begin{cases} 
	g_1(x), & x \in \text{Set } S_1 \\
	g_2(x), & x \in \text{Set } S_2 \\
	.\\
	.\\
	\end{cases}
	\]
\end{defn}
Figure \ref{DPFExample} illustrates an example of a discrete piecewise function. 

The optimal value of a given discrete piecewise function $f(x)$, can be obtained by finding the optimal value from the set of optimal values for each piecewise function, i.e,

\begin{equation*}
f(x)_{opt} = Opt\{g_{1}(x)_{opt} \text{, } g_{2}(x)_{opt}, ...\}
\end{equation*}

\subsection{Properties and Description}\label{SecPropertiesAndDes}

In this section, we first express the harmonization problem as a discrete piecewise optimization problem and highlight its properties. We then describe the discrete piecewise harmonic search algorithm in detail.

In section \ref{BruteForce}, we showed that the multiplier and the base take fixed range of values. The TPE cost function over these ranges can be expressed as follows,
\begin{equation*}
\begin{aligned}
&hcf(m,b,x)= \underset{T'}{\text{Minimize}} \sum_{i=1}^{n} (T_{i} - m * b^{x_i})/T_{i} \forall m,b\\
\end{aligned}
\end{equation*} 

We decompose the above equation over ranges with a fixed multiplier and a range of bases such that that all values of $x_i$ remain identical.
\[ 
hcf(b) = 
\begin{cases} 
g_1(b), & \underset{T'}{\text{Minimize}} \sum_{i=1}^{n} (T_{i} - m_1 * b^{x_i})/T_{i} \\
&\forall 0 < b \le b_{j_1}\\
g_2(b), & \underset{T'}{\text{Minimize}} \sum_{i=1}^{n} (T_{i} - m_1 * b^{x_i})/T_{i} \\
&\forall b_{j_1} < b \le b_{j_2}\\
.\\
.\\
g_{k-1}(b), & \underset{T'}{\text{Minimize}} \sum_{i=1}^{n} (T_{i} - m_1 * b^{x_i})/T_{i} \\
&\forall b_{j_l} < b \le \lfloor T_{n}/m_1 \rfloor \\
g_k(b), & \underset{T'}{\text{Minimize}} \sum_{i=1}^{n} (T_{i} - m_2 * b^{x_i})/T_{i} \\
&\forall 0 < b \le b_{m_1}\\
.\\
.\\
\end{cases}
\]

As can be seen above the harmonization cost function takes the form of a discrete piecewise function. We now find the \emph{local} optimal values of each piecewise function.
  
\begin{lemma}
	\label{lemma4}
	Given a set of harmonic period sets with a fixed multiplier and range of bases, such that all values of $x_i$ remain identical, the harmonic period set with the largest base will be the local (i.e. with the range of bases specified) optima w.r.t the cost functions of rational metrics.
\end{lemma}
\begin{proof}
	Since the multiplier and exponent for all the periods in each set remain constant, the value of each period depends solely on the value of the base. Hence, greater the value of the base, greater the value of the second term, i.e., $b^{x_i}$, and closer the period is to the corresponding period in the original task set resulting in a lower error value. Hence, under the above conditions, any cost function that can be reduced to the form, 
	
	\begin{equation*}
	\begin{aligned}
	& hcf(m,b,x)=\underset{b}{\text{Maximize}}
	& b \\
	\end{aligned}
	\end{equation*}
	
	will indicate that the metric is rational and have a local optimum value at $b = b_{max}$ in the given base range. This can also be thought of as every base $b$, before which the exponent $f_i$ of any term changes, is a locally optimal solution. Consider the cost function used in Figure \ref{BruteForceSearchPlot}.
	\begin{equation*}
	\begin{aligned}
	&hcf(m,b,x)= \underset{T'}{\text{Minimize}} \sum_{i=1}^{n} (T_{i} - m * b^{x_i})/T_{i} \\
	& hcf(b)= \underset{b}{\text{Minimize}} \sum_{i=1}^{n} (T_{i} - m * b^{x_i})/T_{i} \\
	& \text{ Since } m \text{ and } x_i \text{ are constants} \\
	& hcf(b)=\underset{b}{\text{Maximize}} \sum_{i=1}^{n} (m * b^{x_i}) \\
	& hcf(b)=\underset{b}{\text{Maximize}} \text{ } b \\
	\end{aligned}
	\end{equation*} 
\end{proof}
Thus, the total percentage error cost functions' local optimal value will be at $b = b_{max}$ in the given base range. The property of Lemma \ref{lemma4} can be clearly seen in Figure \ref{BruteForceSearchTables}, where the multiplier is fixed at 5, the base $b$ increases from 8 to 22, and $x = \{0,0,1\}$. $T'_{3}$ increases from 40 to 110 which is the highest value it can take for $x_3 = 1$. Hence, under these conditions, in the above example, $b = 22$ will always be the local optimal solution for the $m=5$ and $b=8 \text{ to } 22$. Corresponding results apply to other cost functions of rational metrics too.

\begin{lemma}
	\label{lemma5}
	Given a harmonic period set represented as $m * b^{x_i}$, for given multiplier $m$, the bases where the local optimal solutions occur are  $\lfloor \sqrt[p]{T_{i}/m} \rfloor$ where the range of $p$ is from  $1$ to $(x_i)_{max} $ and $b=1$.
\end{lemma}
\begin{proof}
	From Lemma \ref{lemma2}, we know that for a fixed multiplier $m$, the base varies from 1 to $\lfloor T_{n}/m \rfloor $. From Lemma \ref{lemma3}, we know that, as the base increases, the exponent decreases from $\lfloor log_{2}(T_{i}/m) \rfloor$ to 0. We know, from Lemma \ref{lemma4}, that at every base before which the exponent changes is a local optimal solution. This power flip occurs when the base is just big enough that the exponent has to be deceased. This happen when a base reaches the maximum possible value a given $x_i$ can support.
	\begin{equation*}
	\begin{aligned}
	& T'_{i} = m * b^{x_i} <= T_{i}\\
	& b^{x_i} <= T_{i}/m, m \in \mathbb{N}_{>0}\\
	& b_{max} =
	\begin{cases}
	\lfloor \sqrt[x_i](T_{i}/m) \rfloor, & \text{ if } 0 < x_i <= \lfloor log_{2}T_{i}/m \rfloor \\
	1, & \text{ if } x_i = 0
	\end{cases}\\
	\end{aligned} 
	\end{equation*}
	
	For any value of $x_i$ from 1 to $\lfloor log_{2}T_{i}/m \rfloor$, $\lfloor \sqrt[x_i](T_{i}/m) \rfloor$ will be the set of local optimal bases and, for $x_i = 0$, it is equivalent to $b=1$.  
\end{proof}
\begin{algorithm} [t]
	
	\caption{$GetLocalMinima$}
	\begin{algorithmic}[1]
		\Procedure{getLocalMinima}{}
		\State $m \gets \text{multiplier}$
		\State $T \gets \text{period set from } \tau$
		\State $B \gets 1$
		\For{each $T_{i} \text{ in } T }$
			\For{each $x \text{ from } 1 \text{ to } \lfloor log_{2}(T_{i}/m) \rfloor }$
				\State $B \gets \lfloor \sqrt[x](T_{i}/m) \rfloor$ 
			\EndFor
		\EndFor
		\State \textbf{return} $unique(B)$
		\EndProcedure
	\end{algorithmic}
\end{algorithm}

\begin{algorithm} [t]	
	\caption{Discrete Piecewise Harmonic Search Algorithm}\label{euclid}
	\begin{algorithmic}[1]
		\Procedure{prunedHarmoniSearch}{}
		\State $\tau \gets \text{input task set}$
		\State $T \gets \text{period set from } \tau$
		\State $hcf \gets \text{cost function to minimize}$
		\State $m_{max} = T_{1}$
		\State $error = numMax$
		\For{each multiplier $m \text{ from  } 1 \text{ to } m_{max}$ }
		\State $B = getLocalMinima(m, T)$
		\For{each base $b \text{ in } B$ }
		\State $T'_{temp} = findClosestHarmonicSeries(\tau, m , b) $
		\State $errorTemp = findError(hcf, \tau, T'_{temp})$
		\If {($errorTemp < error) \& (C_{i} <= T'_{temp})$}
		\State $T' = T'_{temp}$
		\EndIf
		\EndFor
		\EndFor
		\State \textbf{return} $T'$ \Comment{Return the closest feasible harmonic period set or null to indicate infeasible result}
		\EndProcedure
	\end{algorithmic}
\end{algorithm}

Algorithm 4 describes the DPHS algorithm in detail. We consider the entire multiplier range from 1 to $T_1$ (from Lemma \ref{lemma1}), but prune the number of bases using Lemma \ref{lemma5}. The $\mathit{GetLocalMinima}$ function takes an input period set and a multiplier value to calculate the local optima. Then, the algorithm only searches the local optima to optimize for the given cost function.

Consider the example in Section \ref{BruteForce}, where $m = 5$ and $x = \{0,0,1\}$ results in a range of bases from 5 to 22. DPHS will search only $b = \lfloor \sqrt[p]{T_{i}/m} \rfloor$, i.e., $b = \lfloor \sqrt[1]{112/5} \rfloor = 22$. As can been seen from Figure \ref{BruteForceSearchTables}, $b=22$ is a local minimum.

It is important to note that priority ordering of the harmonized task set remains identical to the RM priorities of the original task set, if ties are broken in favor of smaller original period values. This is a function of the $FindClosestHarmonicSeries$ method which assigns the harmonic closest to the original period value. This limits all tasks with periods larger than a given harmonic to the value of the harmonic resulting in priorities identical to the RM priorities of the original task set. 

\subsection{Applying the DPHS algorithm}\label{Application}
In the previous section, we described the DPHS algorithm in detail. In this section, we show the benefits of the algorithm by considering real-world applications. We first consider an avionics task set used in Locke et al. \cite{Locke}. Table \ref{GAP} presents this task set.

\begin{table}
	\centering%
\begin{tabular}{ |l|l|l|l|l|}
	\hline
	\multicolumn{5}{ |c| }{Generic Avionic Platform Timing requirements} \\
	\hline
	System & Subsystem  & Task & T(ms) & C(ms)\\ \hline	
	\multirow{5}{*}{Display} 
	& Status Update & $\tau_{10}$ & 200 & 3 \\
	& Keyset & $\tau_{11}$ & 200 & 1 \\
	& Hook Update & $\tau_{7}$ & 80 & 2 \\
	& Graphics Display  & $\tau_{8}$ & 80 & 9\\
	& Stores Update & $\tau_{12}$ & 200 & 1 \\ 
	\hline
	\multirow{1}{*}{RWR} 
	& Contact Mgmt & $\tau_{1}$& 25 & 5 \\
	\hline	
	\multirow{2}{*}{Radar} 
	& Target Update & $\tau_{4}$& 50 & 5 \\
	& Tracking Filter & $\tau_{2}$& 25 & 2 \\
	\hline	
	\multirow{3}{*}{NAV} 
	& Nav Update & $\tau_{6}$& 59 & 8 \\
	& Steering Cmdc & $\tau_{13}$ & 200 & 3 \\
	& Nav Status  & $\tau_{16}$ & 1000 & 1 \\
	\hline	
	\multirow{1}{*}{Tracking} 
	& Target Update  & $\tau_{9}$& 100 & 5 \\
	\hline	
	\multirow{3}{*}{Weapon} 
	& Weapon Protocol  & $\tau_{14}$ & 200 & 1\\
	& Weapon Release  & $\tau_{15}$& 200 & 3 \\
	& Weapon Aim & $\tau_{5}$ & 50 & 3 \\
	\hline
	\multirow{1}{*}{BIT} 
	& Equ. Status Update & $\tau_{17}$ & 1000 & 1 \\
	\hline	
	\multirow{1}{*}{Data Bus} 
	& Poll Bus Devices & $\tau_{3}$ & 40 & 1 \\
	\hline
	\hline	
\end{tabular}
\vspace{0.2cm}
\caption{Generic Avionic Platform Timing requirements \cite{Locke}}
\label{GAP}
\end{table}

\begin{table}
	\centering%
	\begin{tabular}{ |l|l|l|l|l|l|}
		\hline
		\multicolumn{6}{ |c| }{DPHS Outputs} \\
		\hline
		Task & C  & T & $T'_{TSU}$ & $T'_{MPE}$ & $T'_{FOE}$\\ \hline
		$\tau_{1}$  & 5 & 25  &25&20&8 \\
		$\tau_{2}$  & 2 & 25  &25&20&8 \\
		$\tau_{3}$  & 1 & 40  &25&40&40 \\
		$\tau_{4}$  & 5 & 50  &50&40&40 \\
		$\tau_{5}$  & 3 & 50  &50&40&40 \\ 
		$\tau_{6}$  & 8 & 59  &50&40&40 \\
		$\tau_{7}$  & 2 & 80  &50&80&40 \\
		$\tau_{8}$  & 9 & 80  &50&80&40 \\
		$\tau_{9}$  & 5 & 100 &100&80&40 \\
		$\tau_{10}$ & 3 & 200 &200&160&200 \\ 
		$\tau_{11}$ & 1 & 200 &200&160&200 \\
		$\tau_{12}$ & 1 & 200 &200&160&200 \\
		$\tau_{13}$ & 3 & 200 &200&160&200 \\
		$\tau_{14}$ & 1 & 200 &200&160&200 \\
		$\tau_{15}$ & 3 & 200 &200&160&200 \\
		$\tau_{16}$ & 1 & 1000&800&640&1000 \\
		$\tau_{17}$ & 1 & 1000&800&640&1000 \\
		\hline
		&  & Metrics &\cellcolor{green!12}TSU = 97.25\%&\cellcolor{green!12}MPE = 36\%&\cellcolor{green!12}FOE = 213 \\
		\hline
		\hline
	\end{tabular}
	\vspace{0.2cm}
	\caption{DPHS Outputs}
	\label{DPHSOut}
\end{table}

\begin{table}
	\centering%
	\begin{tabular}{ |l|l|l|l|l|l|}
		\hline
		\multicolumn{6}{ |c| }{Comparing Total System Utilization of Output Harmonic Period Sets} \\
		\hline
		Task & C  & $U_{T}$ & $U_{T'_{TSU}}$ & $U_{T'_{MPE}}$ & $U_{T'_{FOE}}$\\ \hline
		$\tau_{1}$  & 5 & 0.2       &0.2		&0.25		&0.625\\
		$\tau_{2}$  & 2 & 0.08      &0.08		&0.1		&0.25\\
		$\tau_{3}$  & 1 & 0.025     &0.04		&0.025		&0.025\\
		$\tau_{4}$  & 5 & 0.1       &0.1		&0.125		&0.125\\
		$\tau_{5}$  & 3 & 0.06      &0.06		&0.075		&0.075\\
		$\tau_{6}$  & 8 & 0.13559   &0.16		&0.2		&0.2\\
		$\tau_{7}$  & 2 & 0.025     &0.04		&0.025		&0.05\\
		$\tau_{8}$  & 9 & 0.1125    &0.18		&0.1125		&0.225\\
		$\tau_{9}$  & 5 & 0.05      &0.05		&0.0625		&0.125\\
		$\tau_{10}$ & 3 & 0.015     &0.015		&0.01875	&0.015\\
		$\tau_{11}$ & 1 & 0.005     &0.005		&0.00625	&0.005\\
		$\tau_{12}$ & 1 & 0.005     &0.005		&0.00625	&0.005\\
		$\tau_{13}$ & 3 & 0.015     &0.015		&0.01875	&0.015\\
		$\tau_{14}$ & 1 & 0.005     &0.005		&0.00625	&0.005\\
		$\tau_{15}$ & 3 & 0.015     &0.015		&0.01875	&0.015\\
		$\tau_{16}$ & 1 & 0.001     &0.00125	&0.0015625	&0.001\\
		$\tau_{17}$ & 1 & 0.001     &0.00125	&0.0015625	&0.001\\
		\hline
		Total U &  & 0.8501 &\cellcolor{green!12}\textbf{0.9725}&\cellcolor{blue!12}1.05313&\cellcolor{blue!12}1.762 \\
		\hline
		\hline
	\end{tabular}
	\vspace{0.2cm}
	\caption{Comparing Total System Utilization of Output Harmonic Period Sets}
	\label{TSU}
\end{table}

\begin{table}
	\centering%
	\begin{tabular}{ |l|l|l|l|l|}
		\hline
		\multicolumn{5}{ |c| }{Comparing Maximum Percentage Error of Output Harmonic Period Sets} \\
		\hline
		Task & C  &$PE_{T'_{TSU}}$ & $PE_{T'_{MPE}}$ & $PE_{T'_{FOE}}$\\ \hline
		$\tau_{1}$  & 5 & 0		&0.2	&0.68\\
		$\tau_{2}$  & 2 & 0		&0.2	&0.68\\
		$\tau_{3}$  & 1 & 0.375	&0		&0\\
		$\tau_{4}$  & 5 & 0		&0.2	&0.2\\
		$\tau_{5}$  & 3 & 0		&0.2	&0.2\\
		$\tau_{6}$  & 8 & 0.153	&0.322	&0.322\\
		$\tau_{7}$  & 2 & 0.375	&0		&0.5\\
		$\tau_{8}$  & 9 & 0.375	&0		&0.5\\
		$\tau_{9}$  & 5 & 0		&0.2	&0.6\\
		$\tau_{10}$ & 3 & 0		&0.2	&0\\
		$\tau_{11}$ & 1 & 0		&0.2	&0\\
		$\tau_{12}$ & 1 & 0		&0.2	&0\\
		$\tau_{13}$ & 3 & 0		&0.2	&0\\
		$\tau_{14}$ & 1 & 0		&0.2	&0\\
		$\tau_{15}$ & 3 & 0		&0.2	&0\\
		$\tau_{16}$ & 1 & 0.2	&0.36	&0\\
		$\tau_{17}$ & 1 & 0.2	&0.36	&0\\
		\hline
		Max $PE$ &  &\cellcolor{blue!12}0.375 &\cellcolor{green!12}\textbf{0.36}&\cellcolor{blue!12}0.68\\
		\hline
		\hline
	\end{tabular}
	\vspace{0.2cm}
	\caption{Comparing Maximum Percentage Error of Output Harmonic Period Sets}
	\label{MPE}
\end{table}

\begin{table}
	\centering%
	\begin{tabular}{ |l|l|l|l|l|}
		\hline
		\multicolumn{5}{ |c| }{Comparing First Order Error of Output Harmonic Period Sets} \\
		\hline
		Task & C  &$FOE_{T'_{TSU}}$ & $FOE_{T'_{MPE}}$ & $FOE_{T'_{FOE}}$\\ \hline
		$\tau_{1}$  & 5 & 0		&5		&17\\
		$\tau_{2}$  & 2 & 0		&5		&17\\
		$\tau_{3}$  & 1 & 15	&0		&0\\
		$\tau_{4}$  & 5 & 0		&10		&10\\
		$\tau_{5}$  & 3 & 0		&10		&10\\
		$\tau_{6}$  & 8 & 9		&19		&19\\
		$\tau_{7}$  & 2 & 30	&0		&40\\
		$\tau_{8}$  & 9 & 30	&0		&40\\
		$\tau_{9}$  & 5 & 0		&20		&60\\
		$\tau_{10}$ & 3 & 0		&40		&0\\
		$\tau_{11}$ & 1 & 0		&40		&0\\
		$\tau_{12}$ & 1 & 0		&40		&0\\
		$\tau_{13}$ & 3 & 0		&40		&0\\
		$\tau_{14}$ & 1 & 0		&40		&0\\
		$\tau_{15}$ & 3 & 0		&40		&0\\
		$\tau_{16}$ & 1 & 200	&360	&0\\
		$\tau_{17}$ & 1 & 200	&360	&0\\
		\hline
		Total $FOE$ &  &\cellcolor{blue!12}484 &\cellcolor{blue!12}1029&\cellcolor{green!12}\textbf{213}\\
		\hline
		\hline
	\end{tabular}
	\vspace{0.2cm}
	\caption{Comparing First Order Error of Output Harmonic Period Sets}
	\label{FOE}
\end{table}

We have already defined the total system utilization (TSU), first order error (FOE), and maximum percentage error (MPE) cost functions in Section \ref{MathRep}. Table \ref{DPHSOut} displays the output harmonic period set from the DPHS algorithm, for each of the cost functions (TSU, FOE, and MPE), for the task set from Table \ref{GAP}. Notice that we obtain a different period set assignment for each cost function. We now look at some properties of the cost functions.

\subsection{The Total System Utilization (TSU) Cost Function}\label{SecTSU}

Let us first consider the TSU (Table \ref{TSU}) cost function. In order to highlight the optimality of the output harmonic period set assignment by the DPHS algorithm, we compare the total system utilization value of the optimal solution ($\sum U_{T'_{TSU}}$) to the total system utilization values of the other harmonic period sets ($\sum U_{T'_{MPE}}$ \& $\sum U_{T'_{FOE}}$). The goal of the TSU cost function is to reduce the total utilization of the resultant harmonic period set. Hence, the period assignment will be biased towards maintaining the periods of tasks with relatively high utilization. This can be evidenced by the fact that the DPHS algorithm leaves the period of the highest-utilization task $\tau_{1}$ unchanged.\\ 
As can be expected, the utilization of the harmonic task set is larger than the original task set.

\begin{table}[t]
	\centering%
	\begin{tabular}{ |l|l|l|}
		\hline
		\multicolumn{3}{ |c| }{DPHS: Running Time Comparison} \\
		\hline
		CF & BFTime(ms)  & PSTime(ms)  \\
		\hline
		TUCF  & 52.115 & 1.473 \\	
		MPEF  & 69.519 & 7.197 \\
		FOCF  & 33.019 & 1.618 \\
		\hline
		\hline
	\end{tabular}
	\vspace{0.2cm}
	\caption{DPHS: Running Time Comparison}
	\label{DPHSRunningTime}
\end{table}

\begin{table}[t]
	\centering%
	\begin{tabular}{ |l|l|l|}
		\hline
		\multicolumn{3}{ |c| }{DPHS: Number of Task Sets Searched Comparison} \\
		\hline
		CF &  BF Task Sets Searched & PF Task Sets Searched\\
		\hline
		$TSU$  & 3806 & 244 \\	
		$MPE$  & 3806 & 244 \\
		$FOE$  & 3806 & 244 \\
		\hline
		\hline
	\end{tabular}
	\vspace{0.2cm}
	\caption{DPHS: Number of Task Sets Searched Comparision}
	\label{DPHSTaskSetsSearched}
\end{table}

\subsection{The Max Percentage Error (MPE) Cost Function}\label{SecMPE}

We now consider the MPE cost function highlighted in Table \ref{MPE}). In order to highlight the optimality of the output harmonic period set assignment by the DPHS algorithm, we compare the maximum percentage error value of the optimal solution ($max(PE_{T'_{MPE}})$) to the maximum percentage error values of the other harmonic period sets ( $max(PE_{T'_{TSU}})$ \& $max(PE_{T'_{FOE}})$). The goal of the MPE cost function is to limit the deviation of the ratio of a particular period to its the original period, so it attempts to affect all periods uniformly.

\subsection{The FOE Cost Function}\label{SecFOE}
We now consider the FOE cost function (Table \ref{FOE}) . Specifically, we compare the total first order error value of the optimal solution ( $\sum FOE_{T'_{FOE}}$ ) to that of the other harmonic period sets ($\sum FOE_{T'_{MPE}}$ \& $\sum FOE_{T'_{TSU}}$). The goal of the FOE cost function is to reduce the distance between the two sets. This can be evidenced by the fact that, as seen in Table \ref{MPE}, the DPHS algorithm ensures the period values of $\tau_{16}$ and $\tau_{17}$ are not reduced as they have the largest period values and the maximum potential for causing first-order errors.

It must be noted that the DPHS algorithm remains the same while computing the local optima and picking the global optimum, irrespective of the rational metric chosen.

Also it is noteworthy that there can exist significantly different optimal solutions based on the choice of the cost function.

\begin{table}[t]
	\centering%
	\begin{tabular}{ |l|l|l|l|l|l|}
		\hline
		\multicolumn{6}{ |c| }{DPHS Outputs Hartstone task set} \\
		\hline
		Task & C  & T & $T'_{TUCF}$ & $T'_{MPEF}$ & $T'_{FOCF}$\\ \hline
		$\tau_{1}$  & 5 & 333.33  &232	&200		&320 \\
		$\tau_{2}$  & 2 & 200  	  &116	&100		&160 \\
		$\tau_{3}$  & 1 & 142.86  &116	&100		&80 \\
		$\tau_{4}$  & 5 & 58.82   &58	&50			&40 \\
		$\tau_{5}$  & 3 & 32.26   &29	&25			&20 \\ 
		\hline
		&  & Metrics &\cellcolor{green!12}U=0.2371&\cellcolor{green!12}39\%&\cellcolor{green!12}145 \\
		\hline
		\hline
	\end{tabular}
	\vspace{0.2cm}
	\caption{DPHS Output Hartstone}
	\label{DPHSHartstone}
\end{table}

We now consider the performance of the DPHS algorithm compared to the brute-force approach (Algorithm \ref{Harmonic Brute-Force Search Algorithm}) for the task set of Table \ref{GAP}. Table \ref{DPHSRunningTime} shows the amount of time each algorithm took to find the optimal solution for our rational metrics. As can be seen, the DPHS algorithm runs much faster than the brute-force approach. Table \ref{DPHSTaskSetsSearched} details the number of task sets each algorithm searched to find the optimal solution. Note that the number of task sets searched are identical for all rational metrics. This is because the total number of tasks searched depends only on the multiplier range (i.e $T_{1}$ Lemma \ref{lemma1})  and base range (i.e $T_{n}$ Lemma \ref{lemma2}). Again DPHS, does far better than the brute-force approach, searching 94\% fewer task sets than the brute-force approach.

We now study the PN Hartstone task set \cite{Hartstone} shown in Table \ref{DPHSHartstone}. Again, the DPHS algorithm produces different optimal assignments as shown in Table \ref{DPHSHartstone}) for different cost functions. The results in Table \ref{DPHSHartstone} reinforce the earlier discussion. 

\begin{figure}[t] 
	\centering
	\includegraphics[page=1, width=8.5cm, height=5.5cm]{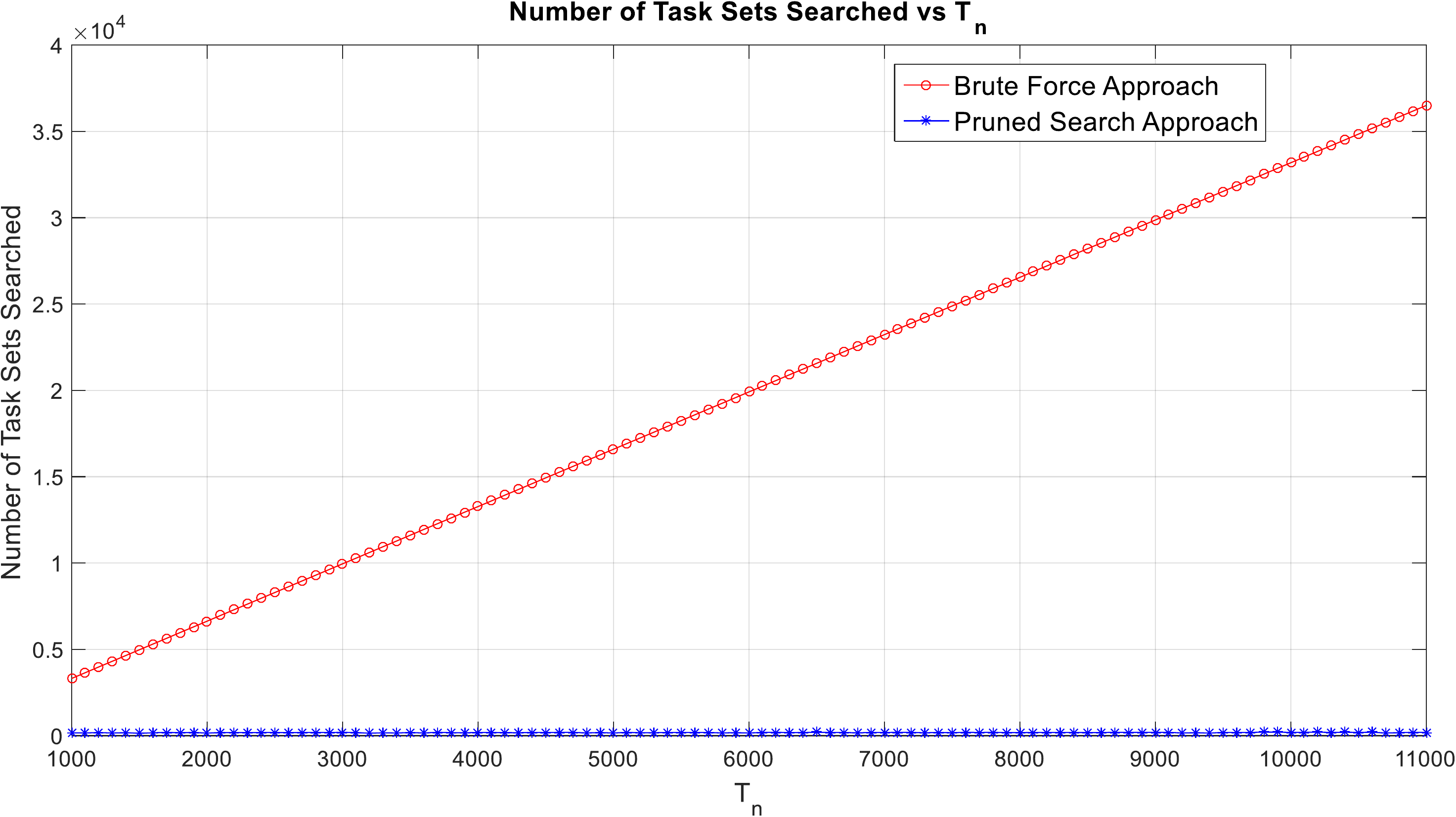}
	\caption{Brute-Force vs DPHS: fixed $T_{1}$ and fixed task set size, vary $T_{n}$}
	\label{Varybase}
\end{figure}

\begin{figure}[t] 
	\centering
	\includegraphics[page=2, width=8.5cm, height=5.5cm]{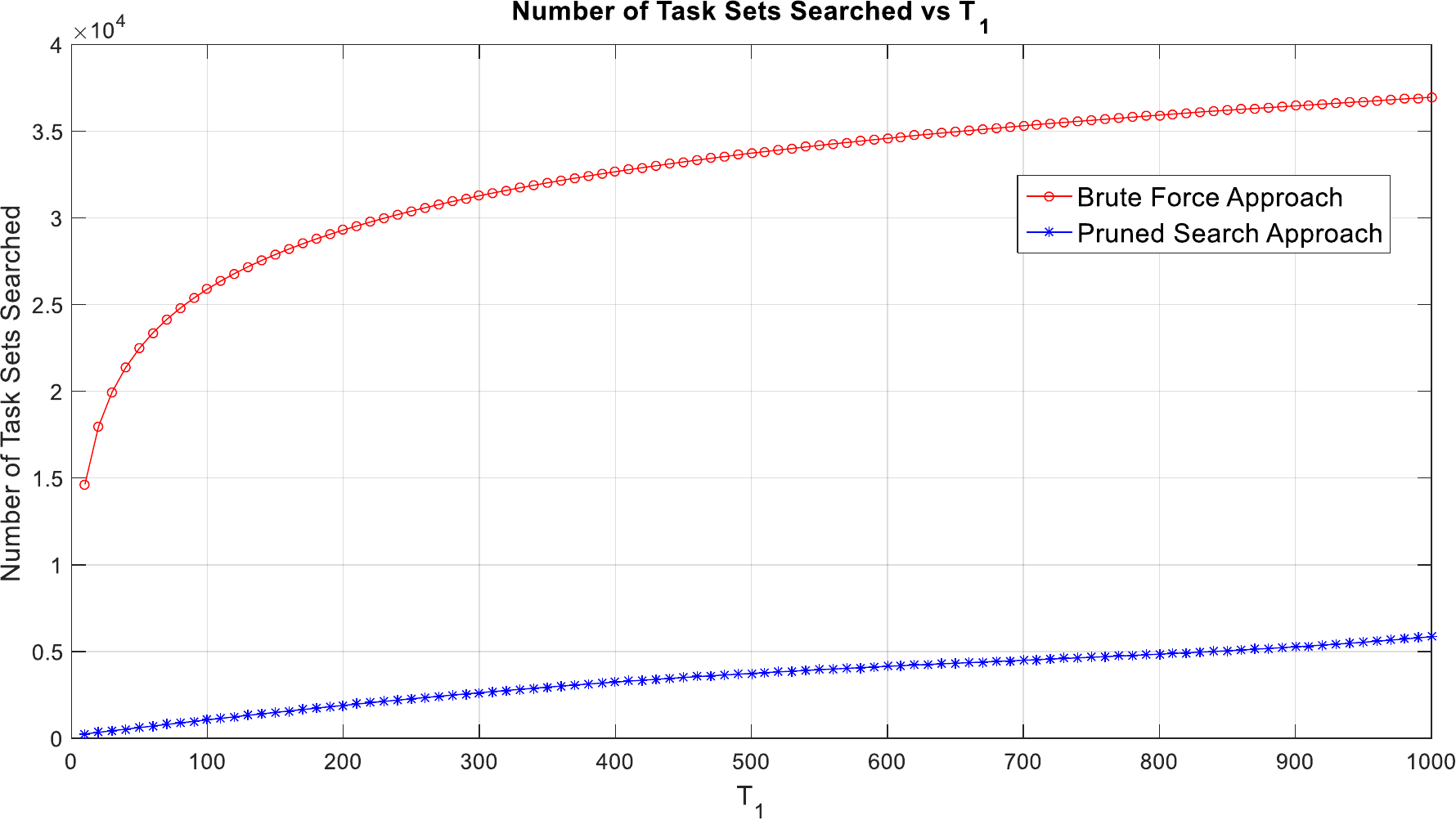}
	\caption{Brute-Force vs DPHS: fixed $T_{n}$ and fixed task set size, varying $T_{1}$}
	\label{Varymultiplier}
\end{figure}

\begin{figure}[t] 
	\centering
	\includegraphics[page=3, width=8.5cm, height=5.5cm]{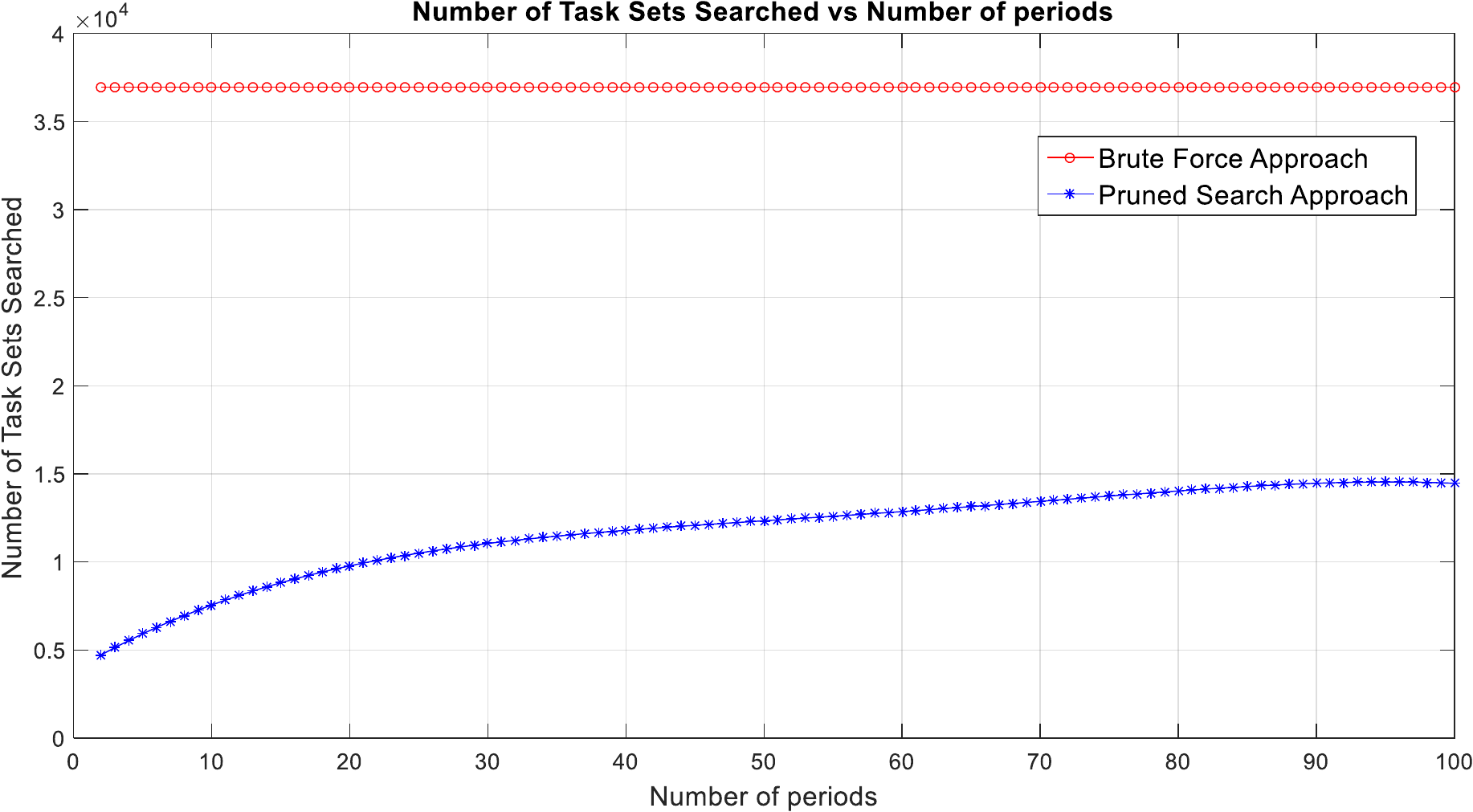}
	\caption{Brute-Force vs DPHS: fixed $T_{1}$ and fixed task set size}
	\label{VaryLength}
\end{figure}

\begin{figure}[t] 
	\centering
	\includegraphics[page=4, width=8.5cm, height=5.5cm]{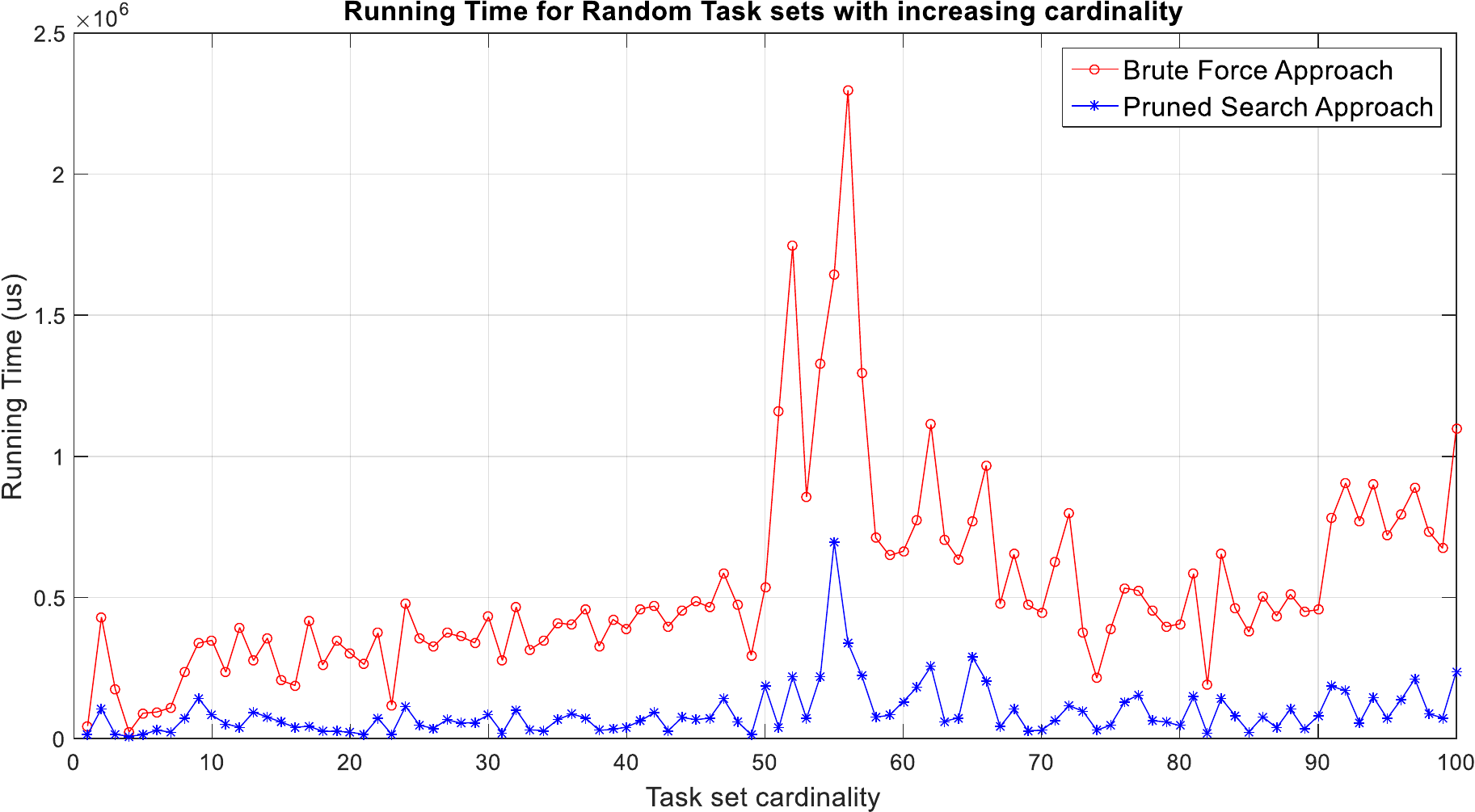}
	\caption{Brute-Force Search Plot}
	\label{RandomRuntime}
\end{figure}

\section{Evaluation and Results}\label{Evaluation}
In the previous sections, we described the DPHS algorithm in detail and showed that it produces the optimal result for any rational metric. We also harmonized some real-world task sets. In this section, we conduct more general evaluation of the performance of the DPHS algorithm against the brute-force search approach with randomly generated values. The first metric is the number of task sets searched.
The DPHS algorithm searches only a small subset of the bases compared to the brute-force approach. The total number of bases decide the number of task sets to be checked. From Lemma \ref{lemma2}, the values the base can take per multiplier range from 1 to $\lfloor T_{n}/m \rfloor$. Hence the number of task sets searched depends on both the value of the multiplier and the value of $T_{n}$. In our experiments, we plot the average number of task sets searched to find the optimal solution over 50,000 randomly generated task sets. In our first experiment, we fix the cardinality of the input task set and the maximum value of the multiplier which is $T_{1}$ (Lemma \ref{lemma1}) to 15, and vary $T_{n}$. We generate random period values that are uniformly distributed in the interval [15, $T_{n}$]. We plot the number of task sets searched. As Figure \ref{Varybase} shows, the number of task sets checked by the brute-force algorithm increase linearly with $T_{n}$. This is because the number of bases searched by the brute-force algorithm per multiplier is equal to $\lfloor T_{n}/m \rfloor$, and, since the multiplier range remains constant, the number of bases is proportional to the $T_{n}$. On the other hand, the DPHS algorithm at max searches only unique bases from the set produced from  $log_{2}(T_{i}/m) + 1$ bases for each term per multiplier, and since the cardinality is fixed, the number of bases searched remains fairly constant as seen in the figure \ref{Varybase}.

In our next experiment, we fix the cardinality of the input task set and the maximum value of the base which is $T_{n}$ (from Lemma \ref{lemma2}) to 5000 and vary $T_{1}$. We generate random period values that are uniformly distributed in the interval [$T_{1}$, 5000]. We again plot the number of task sets searched. As Figure \ref{Varymultiplier} shows, the number of task sets that need to be checked by the brute-force algorithm increases with $T_{1}$. This is because the number of bases searched by the brute-force algorithm per multiplier is equal to $\lfloor T_{n}/m \rfloor$. With the increase in the multiplier range, the number of bases searched per multiplier decreases, explaining the nature of the curve. That is, the increase is more dramatic for smaller $T_{1}$ values and tapers off when $T_{1}$ increases. The result from the DPHS algorithm follows the same trend, but the number of bases searched are much lower because of the pruning.

In our next experiment plotted in figure \ref{VaryLength}, we fix both the range of the multiplier and the base by selecting $T_{1} = 15$ and $T_{n} = 5000$ and vary the cardinality of the input task set. We again generate random period values that are uniformly distributed in the interval [15, 5000]. Since the brute-force approach now has to search the same multipliers and bases every task set, the number of task sets searched remain constant. In the case of the DPHS algorithm since every term in the task set contributes new bases the number of task sets searched increase with the cardinality.

In our final experiment, we measure the running time of each algorithm for task sets, with periods randomly generated periods from a uniform distribution, with increasing cardinality of the input task set. As shown in Figure {\ref{RandomRuntime}, the running time is dominated by values of $T_{1}$ and $T_{n}$ which in this case are random and hence the nature of the plot. But the figure shows that the pruned search does perform significantly better in terms of running time compared to the brute-force algorithm.

In summary, the significant performance improvements of the DPHS over the brute force technique persist across randomly generated tasksets.

\vspace{-0.2cm}
\section{Conclusion}\label{Conclusion}
In this paper, we address the problem of assigning harmonic periods to a task set such that every task gets assigned an integer period less than or equal to its application specified upper bound and the task utilization of every task is less than 1. We presented a mathematical framework to represent integer harmonic task sets in terms of a multiplier, base and the exponent. We derived bounds on the values of these parameters. We presented our solution, the DPHS algorithm, which obtains the provably optimal solution, with excellent run-time performance compared to the brute force approach, for rational metrics. We also showed the benefits of our approach by considering real-world task sets.

\section{Acknowledgements}
We thank Sandeep D'souza and Shunsuke Aoki who were part of the initial discussions.

\end{document}